\documentclass[11pt,reqno]{amsart}
\allowdisplaybreaks[4]
\usepackage{graphicx} %We can use any other package if it is necessary
\usepackage{epsfig}
\usepackage{amssymb}
\usepackage{amsmath}
\usepackage{cite}
\newtheorem{theorem}{Theorem}

\newtheorem{proposition}[theorem]{Proposition}

\newtheorem{lemma}[theorem]{Lemma}
\newcommand{\be}{\begin{equation}}
\newcommand{\ee}{\end{equation}}
\newcommand{\bea}{\begin{eqnarray}}
\newcommand{\eea}{\end{eqnarray}}
\newcommand{\ba}{\begin{array}}
\newcommand{\ea}{\end{array}}
\newcommand{\bean}{\begin{eqnarray*}}
\newcommand{\eean}{\end{eqnarray*}}

\newcommand{\pa}{\partial}

\begin{document}

\title{The integral type gauge transformation and the additional
 symmetry for the constrained KP hierarchy}
\author{Jipeng Cheng\dag, Jingsong He$^*$ \ddag }
\dedicatory {  \dag \ Department of Mathematics, China University of
Mining and Technology, Xuzhou, Jiangsu 221116,
P.\ R.\ China\\
\ddag \ Department of Mathematics, Ningbo University, Ningbo ,
Zhejiang 315211, P.\ R.\ China }

\thanks{$^*$Corresponding author. Email: hejingsong@nbu.edu.cn. Tel: 86-574-87600739}
\begin{abstract}
In this paper, the compatibility between the integral type gauge
transformation and the additional symmetry of the constrained KP
hierarchy is given. And the string-equation constraint in matrix
models is also derived. 
\end{abstract}
\maketitle
%%%%%%%%%%%%%%%%%%%%%%%%%%%%%%%%%%%%%%%%%%%%%%%%%%%%%%
\section{Introduction}
%%%%%%%%%%%%%%%%%%%%%%%%%%%%%%%%%%%%%%%%%%%%%%%%%%%%%%%
The Kadomtsev-Petviashvili (KP) hierarchy (see
\cite{dickey2003,djkm1983} references therein) is one of the most
famous integrable systems, which has many important applications in
theoretical physics and mathematics, such as 2d quantum
gravity\cite{witten1991,kontsevich1992,adler1992} and infinite -
dimensional Lie algebras \cite{kac,jm}. The constrained KP (cKP)
hierarchy\cite{chengli,kss1991,Cheng92}, developed from the point of
view of symmetry constraint, can be viewed as an important reduction
of the KP hierarchy, which includes many well-known integrable
systems, such as the AKNS hierarchy, the Yajima-Oikawa hierarchy and
many others \cite{chengli,kss1991,Cheng92}.

The gauge transformation \cite{chau1992,oevel1993} is an efficient
method to solve the KP hierarchy, which in fact reflects the
intrinsic integrability of the KP hierarchy. Chau {\it et al}
\cite{chau1992} introduce two types of elementary gauge
transformation operators: the differential type $T_D$ and the
integral type $T_I$. By now, the gauge transformations of many
integrable hierarchies related to KP hierarchy have been derived,
for example, the cKP
hierarchy\cite{oevel1993,aratyn1995,chau1997,willox,he2003}, the
constrained BKP and CKP hierarchy\cite{he2007} (cBKP and cCKP), the
discrete KP hierarchy\cite{oevel1996,liu2010}, the $q$-KP
hierarchy\cite{tu1998,he2006} and so on. The additional symmetry
\cite{Fokas1981,Chen1983,OS86,ASM95,D95,takasaki1996,Tu07,he2007b,cheng2011,1li2012,2li2012}
is a kind of symmetry depending explicitly on the space and time
variables, involved in so-called string equation and the generalized
Virasoro constraints in matrix models of the 2d quantum gravity (see
\cite{dickey2003,vM94} and references therein). The corresponding
additional symmetry for the cKP hierarchy is constructed in
\cite{aratyn97,tian2011,tu2011} by an appropriate modification of
the standard additional symmetry flows by adding a set of ``ghost"
symmetry flows. So it is an interesting problem to show the
compatibility between the gauge transformation and the additional
symmetry of the cKP hierarchy, in order to show a new inner
consistency of the integrable hierarchy.

In this paper, we mainly restrict to the integral type gauge
transformation $T_I$, and show that the additional symmetry flows
for the cKP hierarchy commute with the integral type gauge
transformations preserving the form of cKP , up to shifting of the
corresponding additional flows by ordinary time flows, which
reflects one of the intrinsic consistences for the cKP hierarchy:
the compatibility between the integral type gauge transformations
and the additional symmetries. Further, upon the basis of this
result, the string-equation constraint in matrix models is also
derived.

This paper is organized in the following way. In Section 2, we
recall some background on the KP hierarchy. Then the integral type
gauge transformation and the additional symmetry for the cKP
hierarchy are reviewed in Section 3 and Section 4 respectively. In
Section 5, we study the compatibility of the integral type gauge
transformation and the additional symmetry for the cKP hierarchy. At
last, some conclusions and discussions are given in section 6.

%%%%%%%%%%%%%%%%%%%%%%%%%%%%%%%%%%%%%%%%%%%%%%%%%%%%%%
\section{Background on the KP Hierarchy}
%%%%%%%%%%%%%%%%%%%%%%%%%%%%%%%%%%%%%%%%%%%%%%%%%%%%%%%
The KP hierarchy (see \cite{dickey2003,djkm1983} references therein)
is defined as the following Lax equation:
\begin{equation}\label{laxeq}
    \frac{\pa L}{\pa t_n}=[B_n,L],\quad B_n=(L^n)_+,\ n=1,2,3,\cdots
\end{equation}
with the Lax operator given by
\begin{equation}\label{laxop}
    L=\pa+u_2\pa^{-1}+u_3\pa^{-2}+\cdots,
\end{equation}
where $\pa=\frac{\pa}{\pa x}$, coefficient functions $u_i$ are all
the functions of the time variables $t=(t_1=x,t_2,t_3,\cdots)$, and
$(A)_{\pm}$ denote the differential part and the integral part of
the pseudo-differential operator $A$, respectively.

One can also represent the Lax operator in terms of a dressing
operator $W$:
\begin{equation}\label{dressing}
 L=W\pa W^{-1},\qquad W=1+\sum_{j=1}^\infty w_j\pa^{-j},
\end{equation}
and then the Lax equation (\ref{laxeq}) is equivalent to the Sato's
equation:
\begin{equation}\label{satoeq}
\frac{\pa W}{\pa t_n}=-(L^{n})_-W.
\end{equation}

Define next the Baker-Akhiezer (BA) function by:
\begin{equation}\label{wavefunction}
    \psi(t,\lambda)=W(e^{\xi(t,\lambda)})=w(t,\lambda)e^{\xi(t,\lambda)};\
    \ w(t,\lambda)=1+\sum_{i=1}^\infty w_i(t)\lambda^{-i},
\end{equation}
where
\begin{equation}\label{xi}
\xi(t,\lambda)=\sum_{i=1}^\infty t_n\lambda^n.
\end{equation}
Accordingly, there is also an adjoint BA function:
\begin{equation}\label{adjointwavefunction}
    \psi^*(t,\lambda)=W^{*-1}(e^{-\xi(t,\lambda)})=w^*(t,\lambda)e^{-\xi(t,\lambda)};\
    \ w^*(t,\lambda)=1+\sum_{i=1}^\infty w^*_i(t)\lambda^{-i}.
\end{equation}
Then the Lax equation (\ref{laxeq}) can be regarded as the
compatibility conditions for the following system:
\begin{eqnarray}
L(\psi(t,\lambda))=\lambda \psi(t,\lambda),&& \frac{\pa
\psi(t,\lambda)}{\pa t_n}=B_n(\psi(t,\lambda)),\nonumber\\
L^*(\psi^*(t,\lambda))=\lambda \psi^*(t,\lambda),&& \frac{\pa
\psi^*(t,\lambda)}{\pa t_n}=-B_n^*(\psi^*(t,\lambda)).\label{linsys}
\end{eqnarray}
Here for any (pseudo-) differential operator $A$ and a function $f$,
the symbol $A(f)$ will indicate the action of $A$ on $f$, whereas
the symbol $Af$ will denote just operator product of $A$ and $f$,
and $*$ stands for the conjugate operation: $(AB)^*=B^*A^*$,
$\pa^*=-\pa$, $f^*=f$.

The whole KP hierarchy can be characterized by a single function
$\tau(t)$ called $\tau$-function such that
\begin{eqnarray}
w(t,\lambda)&=& \frac{\tau
(t-[\lambda^{-1}])}{\tau(t)},\label{tauepression1}\\
w^*(t,\lambda)&=& \frac{\tau
(t+[\lambda^{-1}])}{\tau(t)},\label{tauexpression2}
\end{eqnarray}
where
$[\lambda^{-1}]=(\lambda^{-1},\frac{1}{2}\lambda^{-2},\cdots)$. This
implies that all dynamical variables $u_i$ in the Lax  operator $L$
can be expressed by $\tau$-function, which is an essential character
of the KP hierarchy..

If the functions $q(t)$ and $r(t)$ satisfy
\begin{equation}\label{eigenfunction}
\frac{\pa q}{\pa t_n}=B_n(q),\quad\quad \frac{\pa r}{\pa
t_n}=-B_n^*(r),
\end{equation}
then we call them the eigenfunction and the adjoint eigenfunction
respectively.

The cKP hierarchy\cite{Cheng92} is defined by restricting the Lax
operator of the ordinary KP hierarchy (\ref{laxeq}) in the following
form:
\begin{equation}\label{ckplax}
    L=\pa+\sum_{i=1}^mq_i\pa^{-1}r_i,
\end{equation}
where $q_i$ and $r_i$ are the eigenfunction and adjoint
eigenfunction respectively.

%%%%%%%%%%%%%%%%%%%%%%%%%%%%%%%%%%%%%%%%%%%%%%%%%%%%%%
\section{The integral type gauge transformation of the cKP hierarchy}
%%%%%%%%%%%%%%%%%%%%%%%%%%%%%%%%%%%%%%%%%%%%%%%%%%%%%%
Let $L^{(0)}$ be the Lax operator of the KP hierarchy (\ref{laxeq}),
and $T$ be a pseduo-differential operator. If the transformation
\begin{equation}\label{gaugegeneral}
    L^{(1)}=TL^{(0)}T^{-1}
\end{equation}
such that
\begin{equation}\label{tranlaxeq}
 \frac{\pa L^{(1)}}{\pa t_n}=[B_n^{(1)},L^{(1)}],\quad B_n^{(1)}=(L^{(1)})^n_+,\ n=1,2,3,\cdots
\end{equation}
still holds for transformed Lax operator $L^{(1)}$, then $T$ is
called the gauge transformation operator of the KP hierarchy.

Chau {\it et al} \cite{chau1992} proved there are the following two
kinds of the gauge transformation operators:
\begin{eqnarray}
{\rm Type\ I:}&& T_D(\chi)=\chi\pa \chi^{-1},\\
{\rm Type\ II:}&& T_I(\mu)=\mu^{-1}\pa^{-1} \mu,
\end{eqnarray}
where $\chi$ and $\mu$ are the eigenfunction and adjoint
eigenfunction respectively. The type I transformation is called the
differential type, while the type II is called the integral type. In
this paper, we mainly study the second one, that is, the integral
type.

The gauge transformation of the cKP hierarchy is investigated in
\cite{oevel1993,aratyn1995,chau1997,he2003}. Here we only review
some results about the integral type. Assume
\begin{equation}\label{ckplax0}
    L^{(0)}=\pa+\sum_{i=1}^mq_i^{(0)}\pa^{-1}r_i^{(0)}
\end{equation}
be the Lax operator of the cKP hierarchy. Under the integral type
gauge transformation $T_I(\mu)$, the transformed Lax operator will
be:
\begin{eqnarray}
  L^{(1)}&=&T_I(\mu)L^{(0)}T_I(\mu)^{-1}=\pa+L^{(1)}_-,\label{ckplax1}\\
  L^{(1)}_-&=&q_0^{(1)}\pa^{-1}r_0^{(1)}+\sum_{i=1}^mq_i^{(1)}\pa^{-1}r_i^{(1)},\label{ckplax1minus}\\
  q_0^{(1)}&=&\mu^{-1},\quad r_0^{(1)}=-(T_D(\mu)L^{(0)*})(\mu),\label{qr01}\\
  q_i^{(1)}&=&T_I(\mu)(q_i^{(0)}),\quad
  r_i^{(1)}=T_I(\mu)^{*-1}(r_i^{(0)})=-T_D(\mu)(r_i^{(0)}),\label{qri1}\\
  \tau^{(1)}&=&\mu \tau^{(0)}.\label{tau1}
\end{eqnarray}
In order to preserve the form (\ref{ckplax0}) of the Lax operator
$L^{(0)}$, $\mu$ is required to coincide with one of the original
adjoint eigenfunctions of $L^{(0)}$, e.g. $\mu=r_1^{(0)}$, since
$r_1^{(1)}=0$ in this case. Applying successive the integral type
gauge transformations
\begin{equation}\label{ckplaxk}
    L^{(k+1)}=T_I^{(k)}L^{(k)}\left(T_I^{(k)}\right)^{-1},\ \
    T_I^{(k)}\equiv\left(r_1^{(k)}\right)^{-1}\pa^{-1}r_1^{(k)}
\end{equation}
yields:
\begin{eqnarray}
  L^{(k+1)}&=&\pa+\sum_{i=1}^mq_i^{(k)}\pa^{-1}r_i^{(k)},\label{ckplaxkminus}\\
  q_1^{(k+1)}&=&\left(r_1^{(k)}\right)^{-1},\quad r_1^{(k+1)}=-(T_D(r_1^{(k)})L^{(k)*})(r_1^{(k)}),\label{qr1k}\\
  q_i^{(k+1)}&=&T_I^{(k)}(q_i^{(k)}),\quad
  r_i^{(k+1)}=\left(T_I^{(k)}\right)^{*-1}(r_i^{(k)}),\ i=2,3,...,m,\label{qrik}\\
  \tau^{(k+1)}&=&r_1^{(k)}\tau^{(k)}.\label{tauk}
\end{eqnarray}
Successive applications of the integral type gauge transformation
will lead to,
\begin{eqnarray}
  r_1^{(n)}&=&(-1)^n\frac{W_{n+1}[r_1^{(0)},\eta_1^{(1)},...,\eta_1^{(n)}]}{W_{n}[r_1^{(0)},\eta_1^{(1)},...,\eta_1^{(n-1)}]},\label{gauger1n}\\
  r_i^{(n)}&=&(-1)^n\frac{W_{n+1}[r_1^{(0)},\eta_1^{(1)},...,\eta_1^{(n-1)},r_i^{(0)}]}{W_{n}[r_1^{(0)},\eta_1^{(1)},...,\eta_1^{(n-1)}]},\label{gaugerin}\\
  \tau^{(n)}&=&r_1^{(n-1)}r_1^{(n-2)}...r_1^{(0)}\tau^{(0)}=W_{n}[\eta_1^{(n-1)},...,\eta_1^{(1)},r_1^{(0)}]\tau^{(0)},\label{gaugetau}
\end{eqnarray}
where $\eta_1^{(k)}=(L^{(0)*})^k(r_1^{(0)})$.

Some useful identities involving the integral type gauge
transformation of the $f\pa^{-1}g$-form are listed in the lemma
below.
\begin{lemma}\label{gaugelemma}
\begin{eqnarray}
T_I(r_a)(M\pa^{-1}r_a)T_I(r_a)^{-1}&=&-r_a^{-1}\pa^{-1}\left\{T_D(r_a)(M\pa^{-1}r_a)^*(r_a)\right\},\label{tiMr}\\
T_I(r_a)(q_a\pa^{-1}N)T_I(r_a)^{-1}&=&-r_a^{-1}\pa^{-1}\left\{T_D(r_a)(q_a\pa^{-1}N)^*(r_a)\right\}+\widetilde{L}(\widetilde{q}_a)\pa^{-1}\widetilde{N},\label{tiqN}\\
T_I(r_a)(M\pa^{-1}N)T_I(r_a)^{-1}&=&-r_a^{-1}\pa^{-1}\left\{T_D(r_a)(M\pa^{-1}N)^*(r_a)\right\}+\widetilde{M}\pa^{-1}\widetilde{N},\label{tiMN}\\
\widetilde{L}^{k+1}(\widetilde{q}_a)&=&T_I(r_a)L^{k}(q_a),\ \ k=0,1,2,....,\label{lqa}\\
(\widetilde{L}^*)^{k-1}(\widetilde{r}_a)&=&T_I(r_a)^{*-1}(L^*)^k(r_a),\
\ k=1,2,3....\label{lra}
\end{eqnarray}
where $r_a$ is one of the adjoint eigenfunctions of the cKP
hierarchy (\ref{ckplax}), $M$ and $N$ are two functions of $t$, and
\begin{eqnarray}
\widetilde{L}=T_I(r_a)LT_I(r_a)^{-1},\ \ \widetilde{q}_a=1/r_a,\ \
\widetilde{M}=T_I(r_a)(M),\ \
\widetilde{N}=T_I(r_a)^{*-1}(N).\label{symbols}
\end{eqnarray}

\end{lemma}
\begin{proof}
Firstly, according to $f\pa^{-1}-\pa^{-1}f=\pa^{-1}f_x\pa^{-1}$ and
$\pa f-f\pa=f_x$,
\begin{eqnarray*}
&&T_I(r_a)(M\pa^{-1}N)T_I(r_a)^{-1}=r_a^{-1}\pa^{-1}r_a(M\pa^{-1}N)r_a^{-1}\pa r_a\\
&=&r_a^{-1}\left(\pa_x^{-1}(r_aM)\pa^{-1}-\pa^{-1}\pa_x^{-1}(r_aM)\right)Nr_a^{-1}\pa
r_a\\
&=&r_a^{-1}\pa_x^{-1}(r_aM)\pa^{-1}(\pa Nr_a^{-1}-(Nr_a^{-1})_x)
r_a-r_a^{-1}\pa^{-1}\left(\pa\pa_x^{-1}(r_aM)Nr_a^{-1}- (\pa_x^{-1}(r_aM)Nr_a^{-1})_x\right)r_a\\
&=&r_a^{-1}\pa_x^{-1}(r_aM)
N-r_a^{-1}\pa_x^{-1}(r_aM)\pa^{-1}(Nr_a^{-1})_x
r_a-r_a^{-1}\pa_x^{-1}(r_aM)N+r_a^{-1}\pa^{-1}
(\pa_x^{-1}(r_aM)Nr_a^{-1})_xr_a\\
&=&T_I(r_a)(M)\pa^{-1}T_I(r_a)^{*-1}(N)-r_a^{-1}\pa^{-1}\left\{T_D(r_a)(M\pa^{-1}N)^*(r_a)\right\},
\end{eqnarray*}
where $\pa_x^{-1}(fg)=\int fg dx$. So (\ref{tiMN}) is be proved.
(\ref{tiMr}) can be derived from (\ref{tiMN}) for $N=r_a$ since
$T_I(r_a)^{*-1}(r_a)=0$.

Then for (\ref{lqa}),
\begin{eqnarray*}
\widetilde{L}^{k+1}(\widetilde{q}_a)&=&T_I(r_a)L^{k+1}T_I(r_a)^{-1}(r_a^{-1})\\
&=& T_I(r_a)L^{k}(\pa+\sum_{i=0}^m q_i\pa^{-1}r_i)r_a^{-1}\pa
r_a(r_a^{-1})= T_I(r_a)L^{k}(q_a).
\end{eqnarray*}
Here we set $q_i\pa_x^{-1}(0)=0$ for $i\neq a$, and
$q_a\pa_x^{-1}(0)=q_a$. (\ref{tiMN}) and (\ref{lqa}) will lead to
(\ref{tiqN}).

At last,
\begin{eqnarray*}
T_I(r_a)^{*-1}(L^*)^k(r_a)=T_I(r_a)^{*-1}(L^*)^{k-1}T_I(r_a)^{*}T_I(r_a)^{*-1}L^*(r_a)=(\widetilde{L}^*)^{k-1}(\widetilde{r}_a).
\end{eqnarray*}

\end{proof}

%%%%%%%%%%%%%%%%%%%%%%%%%%%%%%%%%%%%%%%%%%%%%%%%%%%%%%
\section{Additional symmetries of the cKP hierarchy}
%%%%%%%%%%%%%%%%%%%%%%%%%%%%%%%%%%%%%%%%%%%%%%%%%%%%%%%
The additional symmetry flows \cite{aratyn97} for the cKP hierarchy
(\ref{ckplax}), spanning the Virasoro algebra, are given by:
\begin{equation}\label{addsym}
    \pa_k^*L=[-(ML^k)_-+X_{k-1}^{(1)},L],
\end{equation}
where $M$ is the Orlov-Schulman operator \cite{OS86} defined in the
dressing the ``bare" $M^{(0)}$ operator:
\begin{equation}\label{0osoperator}
    M^{(0)}=\sum_{l\geq1}lt_l\pa^{l-1}=x+\sum_{l\geq1}(l+1)t_{l+1}\pa^l,
\end{equation}
that is,
\begin{eqnarray}
M&=&WM^{(0)}W^{-1}=WxW^{-1}+\sum_{l\geq1}(l+1)t_{l+1}L^l=x+\sum_{l\geq1}(l+1)t_{l+1}(L^l)_++M_-,\label{osoperator}\\
M_-&=&WxW^{-1}-x-\sum_{l\geq1}(l+1)t_{1+l}\frac{\pa W}{\pa
t_l}W^{-1},\label{osminus}
\end{eqnarray}
with (\ref{satoeq}) used in (\ref{osminus}). Define
\begin{equation}\label{xk1}
    X_k^{(1)}=\sum_{i=1}^m\sum_{j=0}^{k-1}\left(j-\frac{1}{2}(k-1)\right)L^{k-1-j}(q_i)\pa^{-1}(L^*)^j(r_i);\
    \ k\geq1,
\end{equation}
which is the essential to ensure the compatibility of the additional
Virasoro symmetry with the constraints (\ref{ckplax}) defining the
cKP hierarchy

Then accordingly, the actions of the additional symmetry flows on
the dressing operators and BA functions are showed that:
\begin{equation}\label{actonwba}
    \pa_k^*W=\left(-(ML^k)_-+X_{k-1}^{(1)}\right)W;\ \
    \pa_k^*\psi(t,\lambda)=\left(-(ML^k)_-+X_{k-1}^{(1)}\right)(\psi(t,\lambda)).
\end{equation}
The corresponding actions on the eigenfunctions $q_i$ and the
adjoint eigenfunctions $r_i$ are derived by considering
$(\pa_k^*L)_-$ listed as follows:
\begin{eqnarray}
\pa_k^*q_i&=&(ML^k)_+(q_i)+\frac{k}{2}L^{k-1}(q_i)+X_{k-1}^{(1)}(q_i),\label{addsymq}\\
\pa_k^*r_i&=&-(ML^k)_+^*(r_i)+\frac{k}{2}(L^*)^{k-1}(r_i)-(X_{k-1}^{(1)})^*(r_i).\label{addsymr}
\end{eqnarray}

%%%%%%%%%%%%%%%%%%%%%%%%%%%%%%%%%%%%%%%%%%%%%%%%%%%%%%
\section{Additional symmetries versus the integral type gauge transformations for the cKP hierarchy}
%%%%%%%%%%%%%%%%%%%%%%%%%%%%%%%%%%%%%%%%%%%%%%%%%%%%%%%
In this section, we will restrict the cKP hierarchy (\ref{ckplax})
to $m=1$ case. And thus its Lax operator is given by
\begin{equation}\label{ckplaxm1}
    L=\pa+q\pa^{-1}r.
\end{equation}

In order to investigate the changes of the additional symmetries
under the integral type gauge transformation $T_I(r)$, some useful
lemmas are needed.

\begin{lemma}
\begin{equation}\label{transformedx}
    T_I(r)X_{k-1}^{(1)}T_I(r)^{-1}=\widetilde{X}_{k-1}^{(1)}
    +\sum_{j=0}^{k-2}\widetilde{L}^{k-j-2}(\widetilde{q})\pa^{-1}\widetilde{L}^j(\widetilde{r})
    +r^{-1}\pa^{-1}\left\{T_I(r)^{*-1}(X_{k-1}^{(1)}-\frac{k}{2}L^{k-1})^*(r)\right\},
\end{equation}
\end{lemma}
\begin{proof}
According to Lemma \ref{gaugelemma}, and (\ref{xk1}) for $m=1$,
\begin{eqnarray*}
&&T_I(r)X_{k-1}^{(1)}T_I(r)^{-1}\\
&=&-r^{-1}\pa^{-1}T_D(r)(X_{k-1}^{(1)})^*(r)
+\sum_{j=1}^{k-2}\left(j-\frac{1}{2}(k-2)\right)\widetilde{L}^{k-j-1}(\widetilde{q})\pa^{-1}(\widetilde{L}^*)^{j-1}(\widetilde{r})\\
&=& -r^{-1}\pa^{-1}T_D(r)(X_{k-1}^{(1)})^*(r)
+\sum_{j=0}^{k-2}\left(j-\frac{1}{2}(k-2)\right)\widetilde{L}^{k-j-2}(\widetilde{q})\pa^{-1}(\widetilde{L}^*)^{j}(\widetilde{r})\\
&&+\sum_{j=0}^{k-2}\widetilde{L}^{k-j-2}(\widetilde{q})\pa^{-1}(\widetilde{L}^*)^{j}(\widetilde{r})
-(1+k-2-\frac{1}{2}(k-2))\widetilde{q}\pa^{-1}(\widetilde{L}^*)^{k-2}(\widetilde{r})\\
&=&
-r^{-1}\pa^{-1}T_D(r)(X_{k-1}^{(1)})^*(r)+\widetilde{X}_{k-1}^{(1)}+\sum_{j=0}^{k-2}\widetilde{L}^{k-j-2}(\widetilde{q})\pa^{-1}\widetilde{L}^j(\widetilde{r})\\
&&-\frac{k}{2}r^{-1}\pa^{-1}T_I(r)^{*-1}(L^*)^{k-1}(r)\\
&=&\widetilde{X}_{k-1}^{(1)}+\sum_{j=0}^{k-2}\widetilde{L}^{k-j-2}(\widetilde{q})\pa^{-1}\widetilde{L}^j(\widetilde{r})+r^{-1}\pa^{-1}\left\{T_I(r)^{*-1}(X_{k-1}^{(1)}-\frac{k}{2}L^{k-1})^*(r)\right\}.
\end{eqnarray*}
\end{proof}
\begin{lemma}
\begin{equation}\label{pakt}
    \pa_k^*T_I(r)\cdot
    T_I(r)^{-1}=r^{-1}\pa^{-1}\left\{T_I(r)^{*-1}\left(-(ML^k)^*_++\frac{k}{2}(L^*)^{k-1}-(X_{k-1}^{(1)})^*\right)(r)\right\}.
\end{equation}
\end{lemma}
\begin{proof}By (\ref{addsymr}),
\begin{eqnarray*}
 \pa_k^*T_I(r)\cdot T_I(r)^{-1}&=&=-T_I(r)\pa_k^*(T_I(r)^{-1})\\
 &=&r^{-1}\pa^{-1}\pa_k^{*}(r) r^{-1}\pa r-r^{-1}\pa_k^*(r)\\
 &=&r^{-1}\pa^{-1}(\pa \pa_k^{*}(r) r^{-1}-(\pa_k^{*}(r) r^{-1})_x)
 r-r^{-1}\pa_k^*(r)\\
 &=& - r^{-1}\pa^{-1}\left\{r(r^{-1}\pa_k^*r)_x\right\}\\
 &=& r^{-1}\pa^{-1}\left\{T_I(r)^{*-1}(\pa_k^*r)\right\}\\
 &=&r^{-1}\pa^{-1}\left\{T_I(r)^{*-1}\left(-(ML^k)^*_++\frac{k}{2}(L^*)^{k-1}-(X_{k-1}^{(1)})^*\right)(r)\right\}.
\end{eqnarray*}
\end{proof}
\begin{proposition}\label{addgau} The additional symmetry flows (\ref{addsym}) for the cKP
hierarchy (\ref{ckplaxm1}) commute with the integral type
transformations preserving the form of cKP , up to shifting of
(\ref{addsym}) by ordinary time flows, that is,
\begin{equation}\label{ckpaddgauge}
    \pa_k^*\widetilde{L}=[-(\widetilde{M}\widetilde{L}^k)_-+\widetilde{X}_{k-1}^{(1)},\widetilde{L}]+\frac{\pa \widetilde{L}}{\pa
    t_{k-1}}.
\end{equation}
\end{proposition}
\begin{proof}Firstly, by (\ref{addsym}),
\begin{eqnarray}
\pa_k^* \widetilde{L}&=& \pa_k^* T_I(r)\cdot
LT_I(r)^{-1}+T_I(r)\pa_k^*L\cdot
T_I(r)^{-1}-T_I(r)LT_I(r)^{-1}\cdot\pa_k^*T_I(r)\cdot T_I(r)^{-1}\nonumber\\
&=&
\left[T_I(r)\left(-(ML^k)_-+X_{k-1}^{(1)}\right)T_I(r)^{-1}+\pa_k^*T_I(r)\cdot
T_I(r)^{-1}, \widetilde{L}\right]\label{proof1}
\end{eqnarray}

Then with the help of (\ref{transformedx}), (\ref{pakt}), and the
following useful formula \cite{or1993}
\begin{equation}\label{usefulformula}
    (r^{-1}\pa^{-1}rPr^{-1}\pa r)_-=r^{-1}\pa^{-1}rP_-r^{-1}\pa
    r-r^{-1}\pa^{-1}\left\{r(r^{-1}(P^*)_+(r))_x\right\},
\end{equation}
we have
\begin{eqnarray}
&&T_I(r)\left(-(ML^k)_-+X_{k-1}^{(1)}\right)T_I(r)^{-1}+\pa_k^*T_I(r)\cdot
T_I(r)^{-1}\nonumber\\
&=&
T_I(r)\left(-(ML^k)_-\right)T_I(r)^{-1}+r^{-1}\pa^{-1}\left\{T_I(r)^{*-1}\left(-(ML^k)^*_+\right)(r)\right\}\nonumber\\
&&+\widetilde{X}_{k-1}^{(1)}
+\sum_{j=0}^{k-2}\widetilde{L}^{k-j-2}(\widetilde{q})\pa^{-1}\widetilde{L}^j(\widetilde{r})\nonumber\\
&=&-(\widetilde{M}\widetilde{L}^k)_-+\widetilde{X}_{k-1}^{(1)}+(\widetilde{L}^{k-1})_-,\label{proof2}
\end{eqnarray}
where the following relation\cite{orlov1996} is used,
\begin{equation}\label{laxminus}
    (\widetilde{L}^{k-1})_-=\sum_{j=0}^{k-2}\widetilde{L}^{k-j-2}(\widetilde{q})\pa^{-1}\widetilde{L}^j(\widetilde{r}).
\end{equation}

At last, the substituting (\ref{proof2}) into (\ref{proof1}) gives
rise to (\ref{ckpaddgauge}).
\end{proof}
\noindent{\bf Remark}: when $m>1$, (\ref{ckpaddgauge}) will not
hold. In fact, when $m>1$, (\ref{transformedx}) will become into
\begin{eqnarray}
&&T_I(r_a)X_{k-1}^{(1)}T_I(r_a)^{-1}\nonumber\\
&=&\widetilde{X}_{k-1}^{(1)}
    +\sum_{j=0}^{k-2}\widetilde{L}^{k-j-2}(\widetilde{q_a})\pa^{-1}\widetilde{L}^j(\widetilde{r_a})
     +r_a^{-1}\pa^{-1}\left\{T_I(r_a)^{*-1}(X_{k-1}^{(1)}-\frac{k}{2}L^{k-1})^*(r_a)\right\},
\end{eqnarray}
where $r_a$ is one of the adjoint eigenfunctions in (\ref{ckplax}).
We can see that
$\sum_{j=0}^{k-2}\widetilde{L}^{k-j-2}(\widetilde{q_a})\pa^{-1}\widetilde{L}^j(\widetilde{r_a})$
can not be written as $(\widetilde{L}^{k-1})_-$ because of
(\ref{laxminus}). Thus from the proof of (\ref{ckpaddgauge}), the
term of $\pa_{t_{k-1}}\widetilde{L}$ in (\ref{ckpaddgauge}) can not
be derived.

 The string-equation constraint can be derived through the
lowest additional symmetry $\pa_0^*$, that is,
\begin{equation}\label{stringeq}
    \pa_0^*L=0\ \ \rightarrow\ \ [M_+, L]=-1;\ \ \pa_0^* r=0\ \ \rightarrow\
    \ M^*_+(r)=0
\end{equation}
With the help of (\ref{eigenfunction}), (\ref{osoperator}) and
(\ref{addsymr}), the constraints on the Lax operator $L$ and the
adjoint eigenfunction $r(t)$ will be derived respectively.
\begin{eqnarray}
\sum_{l\geq1}(l+1)t_{l+1}\frac{\pa L}{\pa
t_l}+[x,L]&=&-1\label{stringlax}\\
\left(-\sum_{l\geq1}(l+1)t_{l+1}\frac{\pa}{\pa
t_l}+x\right)r(t)&=&0\label{stringrt}
\end{eqnarray}
As for the string-equation constraint on the tau functions are
showed in the following proposition:
\begin{proposition} The Wronskian tau functions (\ref{gaugetau}) of the cKP
hierarchy (\ref{ckplaxm1}) generated by the integral type gauge
transformation, invariant under the lowest additional symmetry flow
(\ref{addsym}), satisfy the constraint equation:
\begin{eqnarray}
&&\left(-\sum_{l\geq 1}(l+1)t_{l+1}\frac{\pa}{\pa
    t_l}+nx\right)\frac{\tau^{(n)}}{\tau^{(0)}}\nonumber\\
    &=& \left(-\sum_{l\geq 1}(l+1)t_{l+1}\frac{\pa}{\pa
    t_l}+nx\right)W_n\left[\left(L^{(0)*}\right)^{n-1}(r),\left(L^{(0)*}\right)^{n-2}(r),...,\left(L^{(0)*}\right)(r),r\right]=0\label{stringtau}
\end{eqnarray}
\end{proposition}
\begin{proof}
Firstly, note that $r^{(k)}$ satisfy the same constraint
(\ref{stringrt}) according to Proposition \ref{addgau}. Thus from
(\ref{gaugetau}), we know that
\begin{eqnarray*}
&&\left(-\sum_{l\geq 1}(l+1)t_{l+1}\frac{\pa}{\pa
    t_l}+nx\right)\frac{\tau^{(n)}}{\tau^{(0)}}\\
&=&\left(-\sum_{l\geq 1}(l+1)t_{l+1}\frac{\pa}{\pa
    t_l}+nx\right)r^{(n-1)}r^{(n-2)}...r^{(0)}\\
&=&\sum_{k=0}^{n-1}r^{(n-1)}...\widehat{r^{(n-1)}}...r^{(0)}\left(-\sum_{l\geq
1}(l+1)t_{l+1}\frac{\pa}{\pa
    t_l}+x\right)r^{(k)}=0
\end{eqnarray*}
\end{proof}

%%%%%%%%%%%%%%%%%%%%%%%%%%%%%%%%%%%%%%%%%%%%%%%%%%%%%%
\section{Conclusions and Discussions}
%%%%%%%%%%%%%%%%%%%%%%%%%%%%%%%%%%%%%%%%%%%%%%%%%%%%%%%
The interplay of the integral type gauge transformation $T_I$ with
the additional symmetry at the instance of cKP integrable hierarchy
is showed in Proposition \ref{addgau} (see (\ref{ckpaddgauge})),
which shows the intrinsic coordination of the cKP hierarchy, just
like the compatibility \cite{aratyn97} between the differential type
gauge transformation and the additional symmetry. The
string-equation constraints, through the lowest additional symmetry
$\pa_0^*$, on the Lax operator $L$, the adjoint eigenfunction $r(t)$
and the tau functions, are listed in (\ref{stringlax})
(\ref{stringrt}) and (\ref{stringtau}) respectively. These results
show that KP integrable hierarchy is a kind of integrable system
with many intrinsic coordinations.

 Although this kind of
compatibility between the differential type gauge transformation
$T_D$ and the additional symmetry has been given in \cite{aratyn97}.
But this is not enough to study the cBKP and cCKP because the
combination of $T_D$ and  $T_I$ is necessary to satisfy the
reduction conditions of the BKP and CKP hierarchies\cite{he2007}. So
our results provide a possible basis to explore the consistency of
the gauge transformation and the additional symmetry in the cBKP and
cCKP hierarchies. This work will be done in a near future.

%%%%%%%%%%%%%%%%%%%%%%%%%%%%%%%%%%%%%%%%%%%%%%%%%%%%%%%%%%%%%%
{\bf {Acknowledgements:}}
  {\small  This work is supported by ``the Fundamental Research Funds for the
Central Universities" No. 2012QNA45.}

\end{document}